\documentclass[11pt, a4paper]{article}

\usepackage{graphicx}
\usepackage{amsthm}
\usepackage{amsmath}
\usepackage{amssymb}
\usepackage{comment}
\usepackage{adjustbox}
\usepackage{environ}
\usepackage{tikz}
\usetikzlibrary{arrows,automata,calc, chains, positioning}

\newtheorem{theorem}{Theorem}
\newtheorem{lemma}{Lemma}

\newtheorem{fact}{Fact}
\newtheorem{example}{Example}
\usepackage[margin=2.5cm]{geometry}

\title{Parameterized Complexity of Manipulating Sequential Allocation}
\date{Gran Sasso Science Institute (GSSI), L'Aquila, Italy\\firstname.lastname@gssi.it}
\author{Michele Flammini and Hugo Gilbert}
\newcommand{\AgentSet}{\mathcal{A}}
\newcommand{\ItemSet}{\mathcal{I}}
\newcommand{\PickingSequence}{\pi}
\newcommand{\AllocationFunction}{\phi}
\newcommand{\MSA}{\mathit{ManipSA}}
\newcommand{\rk}{\mathtt{rk}}
\newcommand{\rg}{\mathtt{rg}}

\newcommand\Tau{\mathcal{T}}

\newcommand\DPG{G_{\mathtt{dp}}}
\newcommand\DPV{V_{\mathtt{dp}}}
\newcommand\DPA{A_{\mathtt{dp}}}
\newcommand\DPS{\mathcal{S}_{\mathtt{dp}}}

\NewEnviron{cproblem}[1]{%
\begin{center}\fbox{\parbox{6in}{%
    {\centering\scshape #1\par}%
    \parskip=1ex
    \everypar{\hangindent=1em}%
    \BODY
}}\end{center}}

\begin{document}

\maketitle
\begin{abstract}
    The sequential allocation protocol is a simple and popular mechanism to allocate indivisible goods, in which the agents take turns to pick the items according to a predefined sequence. While this protocol is not strategy-proof, it has been shown recently that finding a successful manipulation for an agent is an NP-hard problem~\cite{aziz2017complexity}. Conversely, it is also known that finding an optimal manipulation can be solved in polynomial time in a few cases: if there are only two agents or if the manipulator has a binary or a lexicographic utility function. 
    In this work, we take a parameterized approach to provide several new complexity results on this manipulation problem. More precisely, we give a complete picture of its parameterized complexity  w.r.t. the following three parameters: the number $n$ of agents, the number $\mu(a_1)$ of times the manipulator $a_1$ picks in the picking sequence, and the maximum range $\rg^{\max}$ of an item. This third parameter is a correlation measure on the preference rankings of the agents. 
    In particular, we show that the problem of finding an optimal manipulation can be solved in polynomial time if $n$ or $\mu(a_1)$ is a constant, and that it is fixed-parameter tractable w.r.t. $\rg^{\max}$ and $n+\mu(a_1)$.  Interestingly enough, we show that w.r.t. the single parameters $n$ and $\mu(a_1)$ it is W[1]-hard.
    Moreover, we provide an integer program and a dynamic programming scheme to solve the manipulation problem and we show that a single manipulator can increase the utility of her bundle by a multiplicative factor which is at most 2.

\end{abstract}
\section{Introduction}

Allocating resources to a set of agents in an efficient and fair manner is one of the most fundamental problems in computational social choice. 
One challenging case is the allocation of indivisible items~\cite{bouveret2016fair,brams2003fair,lipton2004approximately}, e.g., allocating players to teams. 
To address this problem, the sequential allocation mechanism has lately received increasing attention in the AI literature~\cite{aziz2015possible,aziz2016welfare,bouveret2011general,kalinowski2013social,kalinowski2013strategic,levine2012make}. 
This mechanism works as follows: at each time step an agent, selected according to a predefined sequence, is allowed to pick one item among the remaining ones.
Such a protocol has many desirable qualities: it is simple, it can be run both in a centralized and in a decentralized way, and agents do not have to submit cardinal utilities.  
For these reasons, sequential allocation is used in several real life applications, as for instance by several professional sports associations \cite{brams1979prisoners} to organize their draft systems (e.g., the annual draft of the National Basketball Association in the US), and by the Harvard Business School to allocate courses to students \cite{budish2012multi}.

Unfortunately, it is well known that the sequential allocation protocol is not strategy-proof. Stated otherwise, an agent can obtain a better allocation by not obeying her preferences, and this might lead to unfair allocations for the truthful agents \cite{bouveret2011general}. Such a drawback has motivated the algorithmic study of several issues related to strategic behaviors in the sequential allocation setting, the most important one being the computation of a ``successful'' manipulation. Hopefully, if finding a successful manipulation is computationally too difficult, then agents may be inclined to behave truthfully~\cite{walsh2016strategic}.

\paragraph{Related work on strategic behaviors in the sequential allocation setting.}
Aziz et al.~\cite{aziz2017equilibria} studied the sequential allocation setting by treating it as a one shot game. They notably designed a linear-time algorithm to compute a pure Nash equilibrium and a polynomial-time algorithm to compute an optimal Stackelberg strategy when there are two agents, a leader and a follower, and the follower has a constant number of distinct values for the items. If the sequential allocation setting is seen as finite repeated game with perfect information, Kalinowski et al.~\cite{kalinowski2013strategic} showed that the unique subgame perfect Nash equilibrium can be computed in linear time when there are two players. However, with more agents, the authors showed that computing one of the possibly exponentially many equilibria is PSPACE-hard.

Several papers focused on the complexity of finding a successful manipulation for a given agent, also called manipulator. 
Bouveret and Lang~\cite{bouveret2011general} showed that determining if there exists a strategy for the manipulator to get a specific set of items can be done in polynomial time. Moreover, Aziz et al.~\cite{aziz2017complexity} showed that finding if there exists a successful manipulation, whatever the utility function of the manipulator, is a polynomial-time problem. Conversely, the same authors showed that determining an optimal manipulation (for a specific utility function) is an NP-hard problem. Bouveret and Lang~\cite{bouveret2014manipulating} provided further hardness results for finding an optimal manipulation for the cases of non-additive preferences and coalitions of manipulators. 
On the other hand, finding an optimal manipulation can be performed in polynomial time if the manipulator has lexicographic or binary utilities~\cite{aziz2017complexity,bouveret2014manipulating} or if there are only two agents~\cite{bouveret2014manipulating}.  Tominaga et al.~\cite{tominaga2016manipulations} further showed that finding an optimal manipulation is a polynomial-time problem when there are only two agents and the picking sequence is composed of a sequence of randomly generated rounds. More precisely, in each round, both agents get to pick one item and a coin flip determines who picks first.

\paragraph{Our contribution.} We tackle the parameterized complexity of manipulating sequential allocations, and provide a complete picture of the problem w.r.t. the following three parameters: the number $n$ of agents, the number $\mu(a_1)$ of items the manipulator gets to pick in the allocation process, and the maximum range $\rg^{\max}$ of an item, a parameter measuring how close the preference rankings of the agents are. In particular, using a novel dynamic programming algorithm, we show that the problem is in XP with respect to $n$, that is it can be solved in polynomial time if $n$ is constant, and it is  Fixed-Parameter Tractable (FPT) w.r.t. $\rg^{\max}$. Moreover, we show that it is in XP w.r.t. $\mu(a_1)$ and FPT w.r.t. to the sum $n+\mu(a_1)$. Interestingly enough, we prove that the problem is W[1]-hard w.r.t. to the single parameters $n$ and $\mu(a_1)$. As a consequence, our XP results are both tight. Table~\ref{tab:sumup} summarizes our results. Lastly, we provide an integer programming formulation of the problem and show that the manipulator cannot increase the utility of her bundle by a multiplicative factor greater than or equal to 2, this bound being tight.          

\begin{table}[h]
\caption{\label{tab:sumup} \small Our parameterized complexity results on the problem of manipulating sequential allocations.}
    \centering\begin{tabular}{|c|c|c|c|c|}
          \hline
        Parameter &    $n$ & $\mu(a_1)$ & $n + \mu(a_1)$ &  $\rg^{\max}$  \\
               \hline
          Results & In XP and W[1]-hard & In XP and W[1]-hard & In FPT  & In FPT\\
                & Theorems~\ref{theorem : n} and \ref{theorem: hardness on n} & Theorem~\ref{theorem: hardness} & Theorem~\ref{theorem n + mu} & Theorem~\ref{theorem : rg}\\
         \hline
    \end{tabular}
\end{table}

Two results presented in this paper (Theorems \ref{theorem : n} and \ref{theorem : bound}) have independently been found by Xiao and Ling~\cite{DBLP:journals/corr/abs-1909-06747}. Indeed, they present another XP algorithm for parameter $n$ as well as the same bound on the increase in utility that a manipulator can obtain. Interestingly, the proofs and the insights on the picking sequence allocation process that are used are quite different. 
\section{Setting and Notations}
We consider a set \(\AgentSet = \{a_1,\ldots,a_n\}\) of \(n\) agents and a set \(\ItemSet = \{i_1,\ldots,i_m\}\) of \(m\) items. A preference profile \(P = \{\succ_{a_1}, \ldots,\succ_{a_n}\}\) describes the preferences of the agents. 
More precisely, \(P\) is a collection of rankings such that ranking \(\succ_a\) specifies the preferences of agent \(a\) over the items in \(\ItemSet\). The items are allocated to the agents according to the following sequential allocation procedure: at each time step, a picking sequence \(\PickingSequence \in \AgentSet^m\) specifies an agent who gets to pick an item within the remaining ones. Put another way, \(\PickingSequence(1)\) picks first, then \(\PickingSequence(2)\) picks second, and so forth. We assume that agents behave greedily by choosing at each time step their preferred item within the remaining ones. If we view sequential allocation as a centralized protocol, then all agents report their preference rankings to a central authority which mimics this picking process. In the following, w.l.o.g. we use this centralized viewpoint where agents have to report their preference rankings. This sequential process leads to an allocation that we denote by \(\AllocationFunction\). More formally, \(\AllocationFunction\) is a function such that \(\AllocationFunction(a)\) is the set of items that agent \(a\) has obtained at the end of the sequential allocation process. 

\begin{example}[Adapted from Example 1 in~\cite{aziz2017complexity}] \label{example:running}
For the sake of illustration, we consider an instance with 3 agents and 4 items, i.e., \(\AgentSet = \{a_1 , a_2 , a_3\}\) and \(\ItemSet = \{i_1, i_2, i_3, i_4\}\). The preferences of the agents 
are described by the following 
profile:
\begin{align*}
    a_1 &: i_1 \succ i_2 \succ i_3 \succ i_4\\
    a_2 &: i_3 \succ i_4 \succ i_1 \succ i_2\\
    a_3 &: i_1 \succ i_2 \succ i_3 \succ i_4
\end{align*}
and the picking sequence is \(\PickingSequence = (a_1,  a_2,  a_3, a_1)\). Then, \(a_1\) will first pick \(i_1\), then \(a_2\) will pick \(i_3\),  \(a_3\) will pick \(i_2\), and lastly \(a_1\) will pick \(i_4\). Hence, the resulting allocation is given by \(\AllocationFunction(a_1) = \{i_1, i_4\}\), \(\AllocationFunction(a_2) = \{i_3\}\) and \(\AllocationFunction(a_3) = \{i_2\}\).
\end{example}

The allocation \(\AllocationFunction\) is completely determined by the picking sequence \(\PickingSequence\) and the preference profile \(P\). Notably, if one of the agents reports a different preference ranking, she may obtain a different set of items. This different set may even be more desirable to her. Consequently, agents may have an incentive to misreport their preferences.

\begin{example}[Example \ref{example:running} continued]
Assume now that agent \(a_1\) reports the preference ranking \(i_3 \succ i_2 \succ i_1 \succ i_4\). Then she obtains the set of items \(\AllocationFunction(a_1) = \{i_2, i_3\}\). This set of items may be more desirable to \(a_1\) than \(\{i_1, i_4\}\) if for instance items \(i_1\), \(i_2\) and \(i_3\) are almost as desirable as one another, but are all three much more desirable than \(i_4\). 
\end{example}

In this work, we study this type of manipulation. We will assume that agent \(a_1\) is the manipulator and that all other agents behave truthfully. Although the agents are asked to report ordinal preferences, we will assume a standard assumption in the literature that \(a_1\) has underlying additive utilities for the items. More formally, the preferences of \(a_1\) over items in \(\ItemSet\) are described by a set of positive values \(U = \{u(i)|i \in \ItemSet\}\) such that \(i \succ_{a_1} j\) implies \(u(i) > u(j)\). The utility of a set of items \(S\) is then obtained by summation, i.e., \(u(S) = \sum_{i\in S} u(i)\). 
We will denote by \(\succ_T\) the truthful preference ranking of \(a_1\), and by \(\AllocationFunction_{\succ}\) the allocation obtained if agent \(a_1\) reports the preference ranking \(\succ\). Moreover, we will denote by \(u_T\) the value \(u(\AllocationFunction_{\succ_T}(a_1))\) which is the utility value of \(a_1\)'s allocation when she behaves truthfully.  We will say that a preference ranking \(\succ\) is a successful manipulation if \(a_1\) prefers \(\AllocationFunction_{\succ}(a_1)\) to \(\AllocationFunction_{\succ_T}(a_1)\), i.e., if \(u(\AllocationFunction_{\succ}(a_1)) > u_T\).  
Hence, the objective for the manipulator is to find a successful manipulation \(\succ\) maximizing \(u(\AllocationFunction_{\succ}(a_1))\). We are now ready to define formally the problem of \textbf{M}anipulating a \textbf{S}equential \textbf{A}llocation process, called \(\MSA\).  
\begin{cproblem}{\(\MSA\)}
Input: A set \(\AgentSet = \{a_1, a_2, \ldots, a_n\}\) of \(n\) agents where \(a_1\) is the manipulator, a set \(\ItemSet = \{i_1,\ldots,i_m\}\) of \(m\) items, a picking sequence \(\PickingSequence\), a preference profile \(P\) and a set \(U\) of utility values for \(a_1\).

Find: A preference ranking \(\succ\) maximizing \(u(\AllocationFunction_{\succ}(a_1))\).
\end{cproblem}

The \(\MSA\) problem is known to be \(\mathit{NP}\)-hard \cite{aziz2017complexity}. In this work, we will address this optimization problem from a parameterized complexity point of view. We will be mostly interested in three types of parameters: the number of agents, the  number of items that an agent gets to pick, and the range of the items. While the number of agents \(n\) is already clear, let us define the other parameters more formally. We denote by \(\mu(a)\) the number of items that agent \(a\) gets to choose in \(\PickingSequence\) and by \(\mu^{\mathtt{max}}\) the maximum of these values, i.e., \(\mu^{\mathtt{max}} = \max\{\mu(a) | a\in \AgentSet\}\). Let \(\rk_a(i)\) denote the rank of item \(i\) in the preference ranking of agent \(a\). Then, we define the range \(\rg(i)\) of an item \(i\) as: 
\[\rg(i) = \max_{a\in \AgentSet\setminus\{a_1\}} \rk_a(i) - \min_{a\in \AgentSet\setminus\{a_1\}} \rk_a(i)  + 1.\notag \] 
Note that we define the range of an item using only non-manipulators. The maximum range of an item \(\rg^{\mathtt{max}}\) is then defined as \(\max_{i\in \ItemSet} \rg(i)\). 

Let us give some intuitions on parameters $\mu(a_1)$ and \(\rg^{\mathtt{max}}\). 
In the $\MSA$ problem, $\mu(a_1)$ can be seen as a budget parameter for the manipulator. Intuitively, the larger the value of $\mu(a_1)$, the more she can manipulate. It is also the size of a feasible solution, i.e., the size of the bundle \(a_1\) will get. Interestingly, in real-life applications, $\mu(a_1)$ can be much smaller than $|\ItemSet|$ and even much smaller than $|\AgentSet|$. For these reasons, $\mu(a_1)$ is an interesting parameter to study in the $\MSA$ problem. 
On the other hand, parameter \(\rg^{\mathtt{max}}\) measures the correlation between the preferences of the non-manipulators. If \(\rg^{\mathtt{max}}=1\) (its minimal possible value), then all non-manipulators have the same preference ranking. In this case, the manipulation problem becomes easy to solve, as all non-manipulators can be treated as a single agent and the case of two agents is known to be polynomial-time solvable. This simple insight can let us hope that the manipulation problem remains tractable if this parameter is small. An important motivation behind parameter \(\rg^{\mathtt{max}}\) is that, in practice, the preferences of different agents are often correlated.

\section{Positive Parameterized Complexity Results on the $\MSA$ Problem}\label{Sect:positive}
To solve the \(\MSA\) problem, one can simply try the \(m!\) possible preference rankings and see which ones yield the maximum utility value. However, a more clever approach uses the following result.

\begin{fact} [From Propositions 7 and 8 in \cite{bouveret2011general}]
Given a specific set of items \(S\), it is possible to determine in polynomial time whether there exists a preference ranking \(\succ\) such that \(S \subseteq \AllocationFunction_{\succ}(a_1)\). In such case, it is also easy to compute such a ranking.
\end{fact}

Hence, one can try all the \(\binom{m}{\mu(a_1)}\) 
possible sets of \(\mu(a_1)\) items to determine which is the best one that \(a_1\) can get. This approach shows that the \(\MSA\) problem is in XP w.r.t. parameter \(\mu(a_1)\). 

\begin{example}[Example~\ref{example:running} continued] As in this example \(\mu(a_1) = 2\), \(a_1\) just has to determine which one of the \({4 \choose 2}\) following sets she must obtain: \(\{i_1,i_2\}, \{i_1,i_3\}, \{i_1,i_4\}, \{i_2,i_3\}\), \(\{i_2,i_4\}, \{i_3,i_4\}\).
\end{example}

To obtain further positive results, we first design a dynamic programming scheme.  
We will then explain how this scheme entails several positive parameterized complexity results. 
\subsection{A Dynamic Programming Scheme}
Our dynamic programming approach considers pairs \((k,S)\) where \(S \subseteq \ItemSet\) and \(k \in \{0,1,\ldots,\mu(a_1)\}\).  In a state characterized by the pair \((k,S)\), we know that the items in \(S\) have been picked by some agents in \(\AgentSet\) as well as \(k\) other items that have been picked by agent \(a_1\). However, while the items in \(S\) are clearly identified, the identities of these \(k\) other items are unspecified.

Given a pair \((k,S)\), the number of items that have already been picked is \(|S| + k\). Hence, the next picker is \(\PickingSequence(|S| + k + 1)\). Let us denote this agent by \(a\). If \(a = a_1\), i.e., she is the manipulator, then she will pick one more item within the set \(\ItemSet \setminus S\) and we move to state \((k+1,S)\). 
Otherwise, let \(b(a,S)\) denote the preferred item of agent \(a\) within the set \(\ItemSet \setminus S\). Then, two cases are possible: 
\begin{itemize}
    \item In the first case, \(b(a,S)\) has already been picked by agent \(a_1\). Then, we move to state \((k - 1, S \cup \{b(a,S)\})\). Note that this is only possible if \(k\) was greater than or equal to one.
    \item Otherwise \(b(a,S)\) is picked by agent \(a\) and we move to state \((k,S\cup \{b(a,S)\})\).
\end{itemize}

Let us denote by \(V(k,S)\) the maximal utility that the manipulator can get from items in \(\ItemSet\setminus S\) starting from state \((k,S)\). Then the value of an optimal manipulation is given by \(V(k = 0, S = \emptyset)\). From the previous analysis, function \(V\) verifies the following dynamic programming equations: 
\begin{align}
    V(k,S) \!&= \!V(k+1 , S) \text{ if } \PickingSequence(|S| + k + 1) = a_1 \label{eq:dp1}\\
    V(k=0,S) \!&=\! V(k,S\!\cup\!\{b(a,S)\}) \!\text{ if }\! \PickingSequence(|S| + k + 1)\! =\! a \!\neq\! a_1 \label{eq:dp2}\\
    V(k>0,S) \!&=\! \max\big(V(k-1,S\cup\{b(a,S)\}) + u(b(a,S)), \notag\\
          V(k&,S\cup\{b(a,S)\})\big) \text{ if } \PickingSequence(|S| + k + 1) = a \neq a_1 \label{eq:dp3}
\end{align}
where the termination is guaranteed by the fact that \(V(k,S) = \sum_{i \in \ItemSet \setminus S} u(i)\) when \(|S| + k = m\). 

Equations~\ref{eq:dp1}-\ref{eq:dp3} induce a directed acyclic state graph \(\DPG=(\DPV,\DPA)\), where \(\DPV\) is the set of states generated from state \((k=0,S=\emptyset)\) when using these equations and the arcs in \(\DPA\) connect each state to its successor states. We will also denote by \(\DPS = \{S|(k,S)\in\DPV\}\) the set of item-sets \(S\) that are involved in \(\DPV\).

Notably, solving Equations~\ref{eq:dp1}-\ref{eq:dp3} can be performed by building graph \(\DPG\) and running backward induction on it (from the lastly generated states to the initial state \((k=0,S=\emptyset)\)). By standard bookkeeping techniques one can also identify the items that are taken in an optimal manipulation and the order in which they are taken and then recover an optimal ranking to report (where these items are ranked first and in the same order). 

\begin{example}[Example \ref{example:running} continued]
Let us illustrate our approach on our running example. We let \(u(i_1) = 5\), \(u(i_2) = 4\), \(u(i_3) = 3\) and \(u(i_4) = 1\). Figure \ref{fig:Dynprog} displays the resulting state graph \(\DPG\). The values of the states are given next to them and the optimal branches are displayed with thick solid arrows. In this example, \(\DPS = \{\emptyset, \{i_3\}, \{i_1,i_3\},\{i_3,i_4\},\{i_1,i_2,i_3\},\{i_1,i_3,i_4\}\}\). As we can see, the optimal choices for \(a_1\) is to first pick \(i_3\) so that \(i_2\) is still available the second time she becomes the picker. Hence, an optimal manipulation is given by \(\succ = i_3 \succ i_2 \succ i_1 \succ i_4\) which results in an allocation \(\phi_{\succ}(a_1) =\{i_2, i_3\}\) with a utility of $7$ whereas \(u_T = 6\).
\end{example}
\begin{figure}[!t] 
   \centering
    \scalebox{0.9}{\begin{tikzpicture}[->,>=stealth',shorten >=1pt,auto,node distance=7cm,semithick]
  \node[rectangle,draw,text=black] (A0)   at (0,2)  {$(0,\emptyset)$};
  \node[text=black] (A0V)   at (1, 2) {$7$};
  \node[text=black] (A0V)   at (-7, 2) {$\PickingSequence(1) = a_1$};
  \node[rectangle,draw,text=black] (A)   at (0,0.5)  {$(1,\emptyset)$};
  \node[text=black] (AV)   at (1, 0.5) {$7$};
  \node[text=black] (A0V)   at (-7, 0.5) {$\PickingSequence(2) = a_2$};
  \node[rectangle,draw,text=black] (B)   at (-2,-3) {$(1,\{i_3\})$};
  \node[text=black] (BV)   at (-0.9, -3) {$6$};
  \node[text=black] (A0V)   at (-7, -3) {$\PickingSequence(3) = a_3$};
  \node[text=black] (A0V)   at (-7, -6.5) {$\PickingSequence(4) = a_1$};
  \node[rectangle,draw,text=black] (C)   at (1,-1.5) {$(0,\{i_3\})$};
  \node[text=black] (CV)   at (2, -1.5) {$4$};
  \node[rectangle,draw,text=black] (D)   at (-5,-6.5) {$(1,\{i_1,i_3\})$};
  \node[text=black] (DV)   at (-3.8, -6.5) {$5$};
  \node[rectangle,draw,text=black] (E)   at (-2,-5) {$(0,\{i_1,i_3\})$};
  \node[text=black] (EV)   at (-0.5, -5) {$1$};
  \node[rectangle,draw,text=black] (F)   at (3, -3) {$(0,\{i_3,i_4\})$};
  \node[text=black] (FV)   at (4.5, -3) {$4$};
  \node[rectangle,draw,text=black] (G)   at (-5, -8) {$(2,\{i_1,i_3\})$};
  \node[text=black] (GV)   at (-3.8, -8) {$5$};
  \node[rectangle,draw,text=black] (H)   at (-2, -6.5) {$(0,\{i_1,i_2,i_3\})$};
  \node[text=black] (HV)   at (-0.5, -6.5) {$1$};
  \node[rectangle,draw,text=black] (I)   at (-2, -8) {$(1,\{i_1,i_2,i_3\})$};
  \node[text=black] (IV)   at (-0.5, -8) {$1$};
  \node[rectangle,draw,text=black] (J)   at (3, -6.5) {$(0,\{i_1,i_3,i_4\})$};
  \node[text=black] (JV)   at (4.5, -6.5) {$4$};
  \node[rectangle,draw,text=black] (K)   at (3, -8) {$(1,\{i_1,i_3,i_4\})$};
  \node[text=black] (KV)   at (4.5, -8) {$4$};
  \node[text=black,text width = 4cm] (L)   at (-5, -9.5) {$a_1$ has picked $i_2,i_4$, \(u(i_2) + u(i_4) = 5\)};
  \node[text=black,text width = 2.8cm] (M)   at (-2, -9.5) {$a_1$ has picked $i_4$, \(u(i_4) = 1\)};
  \node[text=black,text width = 2.8cm] (N)   at (3, -9.5) {$a_1$ has picked $i_2$, \(u(i_2) = 4\)};

  \path (A0) edge[line width=1.0pt]  node {} (A) 
     (A) edge [dotted,line width=0.5pt,left] node {\(a_2\) picks \(i_3\)} (B)	  
	  (A) edge [line width=1.0pt,text width = 2.8cm, pos = 0.8] node {\(a_1\) has picked \(i_3\), \(u(i_3) = 3\)} (C)
	(B) edge[left,dotted,line width=0.5pt]  node {\(a_3\) picks \(i_1\)} (D)
	(B) edge [line width=1.0pt,text width = 2.8cm]  node {\(a_1\) has picked \(i_1\), \(u(i_1) = 5\)} (E)
	(C) edge [line width=1.0pt] node {\(a_2\) picks \(i_4\)} (F)
		(D) edge [line width=1.0pt] node {} (G)
			(E) edge [line width=1.0pt]  node {\(a_3\) picks \(i_2\)} (H)
				(H) edge [line width=1.0pt] node {} (I)
					(F) edge [line width=1.0pt] node {\(a_3\) picks \(i_1\)} (J)
						(J) edge [line width=1.0pt] node {} (K)
						(G) edge [line width=1.0pt] node {} (L)
						(I) edge [line width=1.0pt] node {} (M)
						(K) edge [line width=1.0pt] node {} (N);
	
	 \path [draw = black, rounded corners, inner sep=100pt,dotted]  
               (-7.8,2.5) 
            -- (4.7,2.5)
	    -- (4.7,1)	 
            -- (-7.8,1)
            -- cycle;
    \path [draw = black, rounded corners, inner sep=100pt,dotted]  
               (-7.8,0.8) 
            -- (4.7,0.8)
	    -- (4.7,-2.3)	 
            -- (-7.8,-2.3)
            -- cycle;
   \path [draw = black, rounded corners, inner sep=100pt,dotted]  
               (-7.8,-2.5) 
            -- (4.7,-2.5)
	    -- (4.7,-6)	 
            -- (-7.8,-6)
            -- cycle;
    \path [draw = black, rounded corners, inner sep=100pt,dotted]  
               (-7.8,-6.1) 
            -- (4.7,-6.1)
	    -- (4.7,-7.5)	 
            -- (-7.8,-7.5)
            -- cycle;
	  \end{tikzpicture}}
    \caption{\label{fig:Dynprog} Directed acyclic state graph \(\DPG\) in Example \ref{example:running}}
\end{figure}
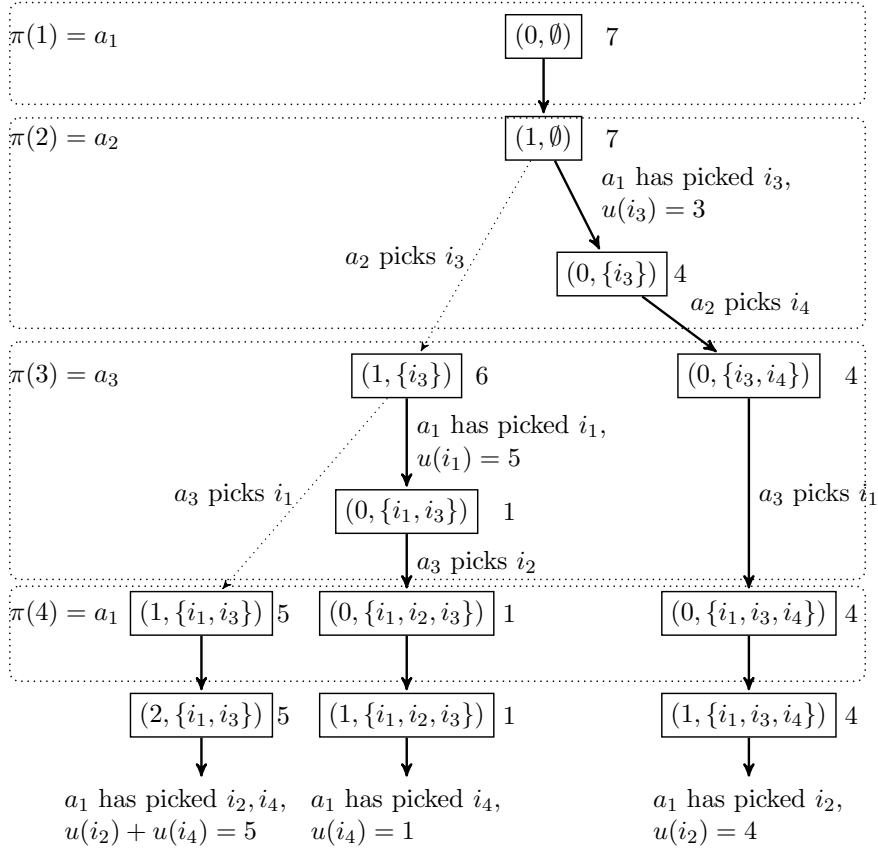
\subsection{Complexity Analysis}
We now provide 
positive parameterized complexity results by proving several upper bounds on \(|\DPV|\). In fact, we will prove bounds on \(|\DPS|\) and use the observation that \(|\DPV| \le (\mu(a_1)+1) |\DPS| \le m |\DPS|\) as there are only \(\mu(a_1)+1\) possible values for $k$ in a state $(k,S)$. 

\paragraph{The algorithm is XP w.r.t. parameter \(n\).} 
Let \(D(a,i)\) denote the set of items that agent \(a\) prefers to item \(i\), i.e., \(D(a,i) = \{j \in \ItemSet | j \succ_a i\}\). Then, for any set \(S\subseteq\ItemSet\), the definition of \(b(a,S)\), which we recall is the preferred element of agent \(a\) in set \(\ItemSet \setminus S\) implies that \(\bigcup_{a \in \AgentSet\setminus\{a_1\}}D(a,b(a,S)) \subseteq S\). Let us denote by \(\Delta\) the set of item-sets for which the equality holds, i.e., \(\Delta = \{S\subseteq\ItemSet | \bigcup_{a \in \AgentSet\setminus\{a_1\}}D(a,b(a,S)) = S \}\). Note that a set \(S\in \Delta\) is completely determined by the vector \((b(a_2,S), \ldots, b(a_n,S))\) and thus \(|\Delta| \leq m^{n-1}\). Our first key insight is that \(\DPS\) is a subset of \(\Delta\).

\begin{lemma} \label{lemma:Delta}
The set \(\DPS\) is a subset of \(\Delta\).
\end{lemma}
\begin{proof}
Let us show by induction that for each state \((k,S)\in \DPV\), \(S\) is in \(\Delta\). 
The result is true for the initial state in which \(S = \emptyset\). Assume that the result is true for a state \((k,S)\). We show that the result also holds for the successor states.  If \(\PickingSequence(|S| + k + 1) = a_1\), then the successor state is \((k+1,S)\) so the result is also true for this new state as \(S\) is unchanged. Otherwise, let \(\PickingSequence(|S| + k + 1) = a^*\), then the successor states are \((k,S\cup\{b(a^*,S)\})\) and \((k-1,S\cup\{b(a^*,S)\}) \). Then we have the following two inclusion relationships:
\begin{align*}
D(a^*,b(a^*,S))\cup\{b(a^*,S)\} \subseteq D(a^*,b(a^*,S\cup\{b(a^*,S)\})),     \\
\forall a \in \AgentSet \setminus \{a_1\}, D(a,b(a,S)) \subseteq D(a,b(a,S\cup\{b(a^*,S)\})).
\end{align*}
 These relationships imply that $S\cup\{b(a^*,S)\}$ is equal to: 
 \begin{align*}
     \bigcup_{a \in \AgentSet\setminus\{a_1\}}\!\!\!\!\!\!D(a,b(a,S)) \cup \{b(a^*,S)\} \subseteq \!\!\!\!\! \bigcup_{a \in \AgentSet\setminus\{a_1\}} \!\!\!\!\!\!D(a,b(a,S \cup \{b(a^*,S)\})),
 \end{align*} 
by the induction hypothesis. As already stated, the reverse inclusion relationship is always true and hence \(S\cup\{b(a^*,S)\} \in \Delta\). 
\end{proof}
Consequently from Lemma~\ref{lemma:Delta}, each state in \(\DPV\) admits two possible representations that we call {\em agent representation} and {\em item representation}. In the item representation, a state \((k,S)\) is represented by a vector of size \(m+1\), i.e., \(S\) is represented by a binary vector of size \(m\). In the agent representation,  a state \((k,S)\) is represented by a vector of size \(n\). In this case, \(S\) is replaced by the vector \((b(a_2,S), \ldots, b(a_n,S))\). Note that processing a state (computing the successor states and the optimal value of the state according to the ones of the successor states) in the agent (resp. item) representation can be done in $O(nm)$ (resp. $O(m)$) operations.
We now show that the \(\MSA\) problem can be solved in polynomial time for any bounded number of agents. 

\begin{theorem} \label{theorem : n}
Problem \(\MSA\) is solvable in \(O(n \cdot m^{n+1})\). As a result, \(\MSA\)  is in XP w.r.t. parameter \(n\). 
\end{theorem}
\begin{proof}
The result follows from the fact that our dynamic programming scheme runs in \(O(n \cdot m^{n+1})\). To obtain this complexity bound, one should use the agent representation. In this case, processing a state requires \(O(nm)\) operations, and one can use a dynamic programming table of size \(m^n\) with one cell per possible vector \((k,b(a_2,S), \ldots, b(a_n,S))\). 
\end{proof}

We now argue that our dynamic programming approach yields an FPT algorithm w.r.t. parameters \(n + \mu(a_1)\), \(n + \rg^{\max}\) and \(\rg^{\max}\) by providing tighter upper bounds on \(|\DPS|\). To use these bounds, we will need the two following lemmata:

\begin{lemma} \label{lemma: building graph}
Under the item representation, the graph \(\DPG\) can be build in $O(m|\DPV|^2)$.
\end{lemma}
\begin{proof}
First note that \(\DPG\) is indeed acyclic. Indeed, given a state that can occur at time step \(t\) of the allocation process (i.e., \(k+|S| = t\)), its successors will either correspond to time step \(t+1\) or will still correspond to time step \(t\) but with a strictly lower value for parameter \(k\). We now show how to incrementally build \(\DPG\) from state \((k = 0,S = \emptyset)\). For each new state generated at the previous iteration, compute its successor states, add edges towards them, and label them with the corresponding utility values. Moreover, each time a state is generated, compare it to the states already generated to avoid the creation of duplicates. If it is indeed a new state, its successors will be computed in the next iteration. This process is repeated until all states are generated. Note that because each state is only processed once, we will generate at most \(2|\DPV|\) states. However, because of the duplicate removal operation performed each time a state is generated, the method runs in $O(m|\DPV|^2)$. Indeed, this step will trigger \(O(|\DPV|^2)\) comparisons, each requiring \(m+1\) operations.  
\end{proof}

\begin{lemma} \label{lemma: solving graph}
Under the item representation, problem \(\MSA\) can be solved in $O(m|\DPV|^2)$.
\end{lemma}
\begin{proof}
By Lemma~\ref{lemma: building graph}, we can build \(\DPG\) in $O(m|\DPV|^2)$ and compute an optimal manipulation by backward induction in $O(m|\DPV|)$.
\end{proof}

\paragraph{The algorithm is FPT w.r.t. parameter \(n + \mu(a_1)\).}
We further argue that \(|\DPS|\) can be upper  bounded by \(m(\mu(a_1)+1)^{n-1}\). This is a consequence of the following lemma, where \(\mu(a,t)\) denotes the number of items that agent \(a\) gets to pick within the \(t\) first time steps.

\begin{lemma} \label{lemma: n + mu}
For each time step \(t\), there is a set \(S_t\) of \(t - \mu(a_1,t)\) items that are always picked within the \(t\) first time steps, whatever the actions of the manipulator.
\end{lemma}
\begin{proof}[Sketch of the proof]
Given an instance \(\mathcal{J}\) of the \(\MSA\) problem, consider the instance \(\mathcal{J}^{-a_1}\) obtained from \(\mathcal{J}\) by removing \(a_1\). Moreover, let us denote by \(S_t^{-a_1}\) the set of items picked at the end of the \(t^{th}\) time step in \(\mathcal{J}^{-a_1}\). This set of size \(t\) is clearly defined as all agents behave truthfully in \(\mathcal{J}^{-a_1}\). We argue that after \(t\) time steps in \(\mathcal{J}\), all items in \(S_{t- \mu(a_1,t)}^{-a_1}\) have been picked whatever the actions of \(a_1\). This can be showed by induction because at each time step where the picker is a non-manipulator, she will pick the same item as in \(\mathcal{J}^{-a_1}\) unless this item has already been picked.
\end{proof}
As a consequence of Lemma~\ref{lemma: n + mu}, for each possible set \(S\) that can appear at time step \(t\) and agent \(a\in \AgentSet\setminus\{a_1\}\), \(b(a,S)\) can only be \(\mu(a_1,t-1)+1\) different items. More precisely, \(b(a,S)\) has to be one of the \(\mu(a_1,t-1)+1\) preferred items of \(a\) in \(\ItemSet\setminus S_{t-1}\). As a result, the number of possible vectors \((b(a_2,S), \ldots, b(a_n,S))\) associated to all the possible sets \(S\) that can appear at time step \(t\) is upper bounded by \((\mu(a_1,t-1)+1)^{n-1}\). Lastly, by using the facts that each set \(S\in \DPS\) is characterized by the vector \((b(a_2,S), \ldots, b(a_n,S))\), that \(\mu(a_1,t) \le \mu(a_1)\) for all \(t\),  and by considering all possible time steps, we obtain that \(|\DPS| \le m(\mu(a_1)+1)^{n-1}\). 
\begin{theorem} \label{theorem n + mu}
Problem \(\MSA\) is solvable in \(O(m^3 (\mu(a_1)+1)^{2n})\).  As a result, \(\MSA\)  is FPT w.r.t. parameter  \(n + \mu(a_1)\).
\end{theorem}
\begin{proof}
This result is a consequence of Lemma~\ref{lemma: solving graph} and the fact that \(|\DPV| \le (\mu(a_1)+1)|\DPS| \le m(\mu(a_1)+1)^{n}\).
\end{proof}

\paragraph{The algorithm is FPT w.r.t. parameter \(n + \rg^{\max}\).} 
We show that \(|\DPS|\) is also upper bounded by \(m(2\rg^{\max})^{n-2}\). This is a consequence of the following lemma. 
\begin{lemma}\label{lem : diff of rk and rg}
For any set \(S \subseteq \ItemSet\), and any two agents \(a_s,a_t \in \AgentSet\setminus\{a_1\}\), 
\[| \rk_{a_s}(b(a_s,S)) - \rk_{a_t}(b(a_t,S))| \leq \rg^{\max} - 1.\]
\end{lemma}
\begin{proof}
If we assume for the sake of contradiction that \(\rk_{a_s}(b(a_s,S)) \ge \rk_{a_t}(b(a_t,S)) + \rg^{\max}\) and  use the fact that \(|\rk_{a_s}(b(a_t,S)) - \rk_{a_t}(b(a_t,S))| < \rg^{\max}\) (by definition of \(\rg^{\max}\)), then we can conclude that \(b(a_t,S) \succ_{a_s} b(a_s,S)\), which contradicts the definition of \(b(a_s,S)\). 
\end{proof}
Lemma~\ref{lem : diff of rk and rg} implies that for each of the \(m\) possible items for \(b(a_2,S)\), there are only \(2\rg^{\max}-1\) possible items for other parameters \(b(a_j,S)\) with \(j>2\). Then, by using the facts that a set \(S\in \DPS\) is characterized by the vector \((b(a_2,S), \ldots, b(a_n,S))\), we obtain that \(\DPS \le m(2\rg^{\max})^{n-2}\). 

\begin{theorem} \label{theorem : n + rg}
Problem \(\MSA\) is solvable in \(O(m^5(2\rg^{\max})^{2(n-2)})\). As a result, \(\MSA\)  is FPT w.r.t. parameter  \(n + \rg^{\max}\). 
\end{theorem}
\begin{proof}
This result is a consequence of Lemma~\ref{lemma: solving graph} and the fact that \(|\DPV| \le m|\DPS| \le m^2(2\rg^{\max})^{n-2}\).
\end{proof}

\paragraph{The algorithm is FPT w.r.t. parameter \(\rg^{\max}\).} 
Lastly, \(|\DPS|\) can also be upper bounded by \(m 2^{2 \rg^{\max}}\). This claim is due to the fact that the set \(S\setminus D(a_2,b(a_2,S))\) cannot contain an item whose rank w.r.t. \(a_2\) is ``too high'', which is proved in the following lemma. 
\begin{lemma} \label{lem : for th rg}
Given \(S\in \DPS\), all \(i\) in \(S\setminus D(a_2,b(a_2,S))\) verify
\[\rk_{a_2}(b(a_2,S)) + 1 \le \rk_{a_2}(i) \le \rk_{a_2}(b(a_2,S)) +2\rg^{\max}.\]    
\end{lemma}
\begin{proof}
First note that by definition \(b(a_2,S)\!\! \not \in\!\! S\) (because \(b(a_2,S)\!\! \in\!\!\ItemSet\setminus S\)) and \(D(a_2,b(a_2,S)) \!\!= \!\!\{i\in \ItemSet | \rk_{a_2}(i) < \rk_{a_2}(b(a_2,S))\}\). Hence, the first inequality of the lemma hold.

Let us assume for the sake of contradiction that there exists \(i\!\in\! S\setminus D(a_2,b(a_2,S))\) such that \(\rk_{a_2}(i) \!>\! \rk_{a_2}(b(a_2,S)) +2\rg^{\max}\). 
Because \(S\) belongs to \(\Delta\), we have that \(S\setminus D(a_2,b(a_2,S)) = \bigcup_{a \in \AgentSet\setminus\{a_1\}}D(a,b(a,S))\setminus D(a_2,b(a_2,S))\). Hence, there exists \(a_j\) with \(j\ge 3\) such that \(i\in D(a_j,b(a_j,S))\). 
By definition of \(\rg^{\max}\), we have that \(\rg^{\max} > \rk_{a_2}(i) - \rk_{a_j}(i)\), or equivalently that \(\rk_{a_j}(i) > \rk_{a_2}(i) - \rg^{\max}\), which yields that \(\rk_{a_j}(b(a_j,S)) > \rk_{a_j}(i) > \rk_{a_2}(i) - \rg^{\max} > \rk_{a_2}(b(a_2,S)) +\rg^{\max}\). This contradicts Lemma~\ref{lem : diff of rk and rg}.
\end{proof}
As a consequence of Lemma~\ref{lem : for th rg}, \(|\DPS|\) is upper bounded by \(m 2^{2 \rg^{\max}}\) because there are at most \(m\) possible items for \(b(a_2,S)\), and for each of them, there are at most \(2^{2 \rg^{\max}}\) possible sets for \(S\setminus D(a_2,b(a_2,S))\).
\begin{theorem} \label{theorem : rg}
Problem \(\MSA\) is solvable in \(O(m^5 2^{4 \rg^{\max}})\). As a result, \(\MSA\)  is FPT w.r.t. parameter  \( \rg^{\max}\). 
\end{theorem}
\begin{proof}
This result is a consequence of Lemma~\ref{lemma: solving graph} and the fact that \(|\DPV| \le m|\DPS| \le m^2 2^{2 \rg^{\max}}\).
\end{proof}

\textit{Remark.} Note that it is easy to prove that, in contrast, the problem is NP-hard even if the average range of the items is of 2.\footnote{ One can use a reduction with a sufficiently large number of dummy items ranked last and in the same positions by all agents.} Furthermore, Theorem \ref{theorem : n + rg} might seem less appealing as the \(\MSA\) problem is FPT w.r.t. parameter \(\rg^{\max}\) alone. However, we would like to stress that the time complexity of Theorem \ref{theorem : n + rg} might be more interesting than the one of Theorem \ref{theorem : rg} for a small number of agents.

\section{Hardness Results on the \(\MSA\) Problem} \label{Sect:MU}

We have seen in the last section that the \(\MSA\) problem is in XP w.r.t. parameters \(n\) and \(\mu(a_1)\) and that it is in FPT for parameter \(\rg^{\max}\). One could hope for more positive results for parameters \(n\) and \(\mu(a_1)\), as an FPT algorithm. However, we show in this section that the \(\MSA\) problem is W[1]-hard w.r.t. each of these two parameters.\\ 

We start with the hardness result on parameter \(\mu(a_1)\). In fact, we obtain a stronger result by proving that even determining if there exists a successfully manipulation is W[1]-hard w.r.t. \(\mu^{\max}\). Note that, by definition, \(\mu^{\max}\) is greater than or equal to \(\mu(a_1)\). 
\begin{theorem} \label{theorem: hardness}
Determining if there exists a successful manipulation for \(a_1\) is W[1]-hard w.r.t. parameter \(\mu^{\mathtt{max}}\). 
\end{theorem}

\begin{proof}
We make a parameterized reduction from the CLIQUE problem where given a graph \(G = (V,E)\) and an integer \(k\), we wish to determine if there exists a clique of size \(k\). This problem is W[1]-hard w.r.t. parameter \(k\). W.l.o.g., we make the assumptions that \(|V| > k\) and that \(|E| > k(k-1)/2\) (because otherwise it is trivial to determine if there is a clique of size \(k\)).\\ 

From an instance of CLIQUE, we create the following \(\MSA\) instance.\\

\emph{Set of items.} We create two items, \(g_{\{i,j\}}\) (a good item) and \(w_{\{i,j\}}\) (one of the worst items), for each edge \(\{i,j\}\in E\) and two items, \(b_{i}\) (one of the best items) and \(m_i\) (a medium item), for each vertex \(i\in V\). Put another way, \(\ItemSet = \{g_{\{i,j\}},w_{\{i,j\}}|\{i,j\}\in E\}\cup\{b_i,m_i| i \in V\}\) and the number of items is thus \(|I| = 2|V| + 2|E|\).\\

\emph{Set of agents.} We create one agent \(e_{\{i,j\}}\) for each edge \(\{i,j\}\in E\) and one agent \(v_{i}\) for each vertex \(i\in V\). The top of \(e_{\{i,j\}}\)'s ranking is \(g_{\{i,j\}} \succ m_i \succ m_j \succ w_{\{i,j\}}\) (which one of \(m_i\) or \(m_j\) is ranked first can be chosen arbitrarily). The top of \(v_{i}\)'s ranking is \(b_{i} \succ m_i\) . We also create \(|V| - k - 1\) agents \(c_t\) for \(t\in \{1,\ldots, |V| - k - 1\}\) (whose role is to collect medium items) such that the top of the ranking of each \(c_t\) is \(m_1\succ m_2\ldots \succ m_{|V|}\). Last, the manipulator, that we denote by \(a_1\) to be consistent with the rest of the paper,
has the following preferences: he first ranks items \(b_{i}\), 
then items \(g_{\{i,j\}}\), then items \(m_i\) and last items \(w_{\{i,j\}}\). To summarize, \(\AgentSet = \{e_{\{i,j\}}|\{i,j\}\in E\}\cup\{v_i|i\in V\}\cup\{c_t|t\in \{1,\ldots, |V| - k - 1\}\}\cup\{a_1\}\) and there are \(|\AgentSet| = 2|V|+|E|-k\) agents.\\

\emph{Picking sequence.} The picking sequence \(\PickingSequence\) is composed of the following rounds: 
\begin{itemize}
    \item Manipulator round 1:  \(a_1\) gets to pick \(k\) items.
    \item Vertex round: each agent \(v_{i}\) gets to pick one item.
    \item Manipulator round 2:  \(a_1\) gets to pick \(k(k-1)/2\) items.
    \item Edge round: each agent \(e_{\{i,j\}}\) gets to pick one item.
    \item Medium item collectors round: each agent \(c_t\) gets to pick one item.
    \item Manipulator round 3:  \(a_1\) gets to pick one item.
    \item End round: the remaining items can be shared arbitrarily within the non-manipulators such that each of them gets at most one new item. 
\end{itemize}
Note that \(\mu^{\mathtt{max}} = \mu(a_1) = \frac{k(k+1)}{2} + 1\).\\

\emph{Utility values of \(a_1\).} For ease of presentation, we will act as if there were only four different utility values, even if \(a_1\) is asked to report a complete preference order. One can remove this assumption, making preferences strict, by using sufficiently small \(\epsilon\) values. 
In this sketch of proof, each item \(b_{i}\) has utility \(4\). Each item \(g_{\{i,j\}}\) has utility \(3\). Each item \(m_i\) has utility \(2\). Lastly, items \(w_{\{i,j\}}\) have utility \(1\). In this simplified setting, we set \(\succ_T\) as being one specific ranking consistent with the utility values of \(a_1\) and we wish to determine if there exists another ranking yielding a strictly higher utility.\\

\emph{Sketch of the proof.} At the end of the vertex round, all the best items are gone, as they have already been picked by \(a_1\) or by the vertex agents $v_i$. Similarly, at the end of the edge round, none of the good items are left. Hence, at the third manipulator round, when \(a_1\) picks her last item, the best she can hope for is a medium item. Consequently, the maximum utility she might achieve is accomplished by picking $k$ best items in her first round, $k(k-1)/2$ good items in her second round, and finally a medium item in her third round, for an overall utility of $4k + 3k(k-1)/2 + 2$. Note that she can always pick any set of \(k\) best items in her first round and then (whatever the previous \(k\) best items) pick any set of \(k(k-1)/2\) good items in her second round. Hence obtaining an overall utility of $4k + 3k(k-1)/2 + 1$ is always possible. Note also that, if $\{b_{i_1},\ldots,b_{i_k}\}$ are the $k$ best items selected by \(a_1\) at the first round, then in the vertex round the vertex agents $\{v_{i_1},\ldots,v_{i_k}\}$ will pick the medium items $\{m_{i_1},\ldots,m_{i_k}\}$. Moreover, before the third manipulator round, agents $c_t$ will pick additional $|V|-k-1$ medium items. So, a medium item is left at the third manipulator round only if none of the edge agents picks a medium item in the edge round. According to her preference ranking, any such agent $e_{\{i,j\}}$ will not pick a medium item iff $g_{\{i,j\}}$ is still available, or if $g_{\{i,j\}}$, $m_i$ and $m_j$ have all already been picked. 
If $g_{\{i,j\}}$ is one of the $k(k-1)/2$ good items that have been already picked at the manipulator second round, then $m_i$ and $m_j$ have already been picked before by $v_i$ and $v_j$
in the vertex round, if $b_i$ and $b_j$ were already taken in the first manipulator round. In conclusion, none of the medium items are picked by the edge agents iff the $k(k-1)/2$ edges $e_{\{i,j\}}$ for which $g_{\{i,j\}}$ has already been picked at the second manipulator round have as endpoints only nodes in $\{v_{i_1},\ldots,v_{i_k}\}$, and this is possible iff $\{v_{i_1},\ldots,v_{i_k}\}$ forms a clique in the initial graph $G$. Summarizing, there exists a strategy for \(a_1\) achieving an overall utility of $4k + 3k(k-1)/2 + 2$ iff $G$ has a clique of $k$ nodes. 
It remains to show that we could solve the CLIQUE problem if we could determine if there exists a successful manipulation. This fact results from the following disjunction of two cases: If \(u_T = 4k + 3k(k-1)/2 + 2\), then we can conclude that there exists a clique of size \(k\); Otherwise, if \(u_T = 4k + 3k(k-1)/2 + 1\), then there exists a clique of size \(k\) iff there exists a successful manipulation for \(a_1\). 
\end{proof}

\textit{Remark:} Aziz et al.~\cite{aziz2017complexity} considered a sequential allocation setting in which the manipulator has a binary utility function but is asked to provide a complete preference order. In this setting, the manipulation problem consists in finding a ranking maximizing the utility of the bundle she gets. While the authors showed that this problem can be solved in polynomial time, the reduction used in the sketch of the proof of Theorem~\ref{theorem: hardness} shows that this problem is NP-hard if the manipulator has a utility function involving four different values (instead of two).\\

 Similarly, we obtain that the \(\MSA\) problem is W[1]-hard w.r.t. the number of agents.

\begin{theorem} \label{theorem: hardness on n}
$\MSA$ is W[1]-hard w.r.t. the number of agents. 
\end{theorem}
\begin{proof}
We design a parameterized reduction from MULTICOLORED CLIQUE. In this problem, given a graph \(G = (V,E)\) with vertex set $V=\{v_1,\ldots,v_n\}$, an integer \(k\), and a vertex coloring \(\phi: V \rightarrow \{1,\ldots,k\}\), we wish to determine if there exists a clique of size \(k\) in \(G\) containing exactly one vertex per color. MULTICOLORED CLIQUE is known to be W[1]-hard w.r.t. parameter \(k\) \cite{fellows2009parameterized}. \\

\emph{Idea of the proof}: We resort on the nice mathematical tool of Sidon sequences. These sequences associate to each number $i$ in $\{1,\ldots, n\}$ a value $id(i)$ such that, for every pair $(i,l)$ with \(i\leq l\), the sum $id(i)+ id(l)$ is different from the one of any other different pair of elements in $\{1,\ldots, n\}$. We use the construction of Erd\"os and Tur\`an \cite{erdos1941problem}, by setting \(id(i) = 2pi + (i^2 \mod p)\) for every \(i \in \{1,\ldots,n\}\), where $p$ is the smallest prime number greater than $n$. Notice that, by the Bertrand-Chebyshev theorem \cite{chebyshev1852memoire}, \(p < 2n\), and thus $id(i)=O(n^2)$. This sequence will be used in the following way. We create a large set of items $B_j$ for each color $j$. In the first picking round, the manipulator will be able to pick a large number of items within these sets. To recover a solution of the MULTICOLORED CLIQUE problem, we will show that, if a multicolored clique $\{v_{i_1},\ldots,v_{i_k}\}$ exists in which each vertex  $v_{i_j}$ has color $\phi(v_{i_j})=j$, then in an optimal manipulation the manipulator should pick exactly $(k+1) \cdot id(i_j)$ items in each set $B_j$. The edges of the clique will then be identified by the sums $id(i_j)+id(i_r)$ for all pairs of vertices $\{v_{i_j},v_{i_r}\} \subset \{v_{i_1},\ldots,v_{i_k}\}$. \\  

From a MULTICOLORED CLIQUE instance \((G=(V,E),k,\phi)\), we construct the following \(\MSA\) instance.\\

\emph{Set of items}:
\begin{itemize}
    \item For each color $j$, we create a set $B_j$ of \((k+1)\cdot id(n)\) items, and two sets $Id_j$ and $Id_{\overline{j}}$ of $id(n)+2$ items each.  
    The purpose of items in $Id_j \cup Id_{\overline{j}}$ is to ensure that the number of items picked by \(a_1\) in \(B_j\) is of the form $(k+1) \cdot id(i)$ such that $\phi(v_i) = j$. 
    \item For each pair of colors $\{j,r\}$ with $j \neq r$, we create a set $Id_{\{j,r\}}$ of $2 \cdot (id(n) +1)$ items. 
    The purpose of  items in $Id_{\{j,r\}}$ is to ensure that, whenever \(a_1\) picks $(k+1) \cdot id(i)$ items in $B_j$ and $(k+1) \cdot id(l)$ items in $B_r$ for two given vertices $v_i$ and $v_l$ of colors $\phi(v_i) = j$ and $\phi(v_l) = r$, then $\{v_i,v_l\}$ is an edge of $G$. 
    \item We create a set $D$ of $k(k+1)\cdot id(n)$ items. 
    In a first picking round, the manipulator will be able to pick items in $\bigcup_{j=1}^k B_j \cup D$. The purpose of items in $D$ is to make it possible for \(a_1\) to adjust the number of items she picks in $\bigcup_{j=1}^k B_j$.  
    \item Last, we add a set $Z$ of $2k(k+1)id(n)$ items. 
    Set $Z$ will be used as a buffer of items where each non-manipulator will pick when no items in $Id_j$, $Id_{\overline{j}}$, or $Id_{j,r}$ are left, so as to avoid mutual conflicts.
\end{itemize}

\emph{Set of agents}: We create two agents $c_j$ and $\overline{c}_j$ per color $j$, and two agents $p_{j,r}$ and $p_{r,j}$ for each pair of colors $\{j,r\}$ such that \(j\neq r\). Moreover, we create one agent denoted by $d$ and one manipulator $a_1$. In total, there are $k(k+1)+2$ agents. We now detail the top of the preference rankings of non-manipulators, where by abuse of notations, we use $S\succ S'$ to denote the fact that items in $S$ are ranked before the ones in $S'$, while the order inside each set is indifferent.
\begin{itemize}
    \item Agent $c_j$, for  $1 \leq j \leq k$: $B_j \succ Id_j \succ Z \succ \ldots$.
    \item Agent $\overline{c}_j$, for $1 \leq j \leq k$: $B_j \succ Id_{\overline{j}} \succ Z \succ \ldots$.
    \item Agent $p_{j,r}$, for $1 \leq j \neq r \leq k$: $B_j \succ  Id_{\{j,r\}} \succ Z \succ \ldots.$
    \item Agent $p_{r,j}$ for $1 \leq j \neq r \leq k$: $B_r \succ Id_{\{j,r\}} \succ Z \succ \ldots.$
    \item Agent $d$: $D \succ Z \succ \ldots$.
\end{itemize}
As an important remark, notice that agents $p_{j,r}$ and $p_{r,j}$ rank items in $Id_{\{j,r\}}$ identically.\\

\emph{Picking sequence}: \(\PickingSequence\) is composed of the following rounds: 
\begin{itemize}
    \item Manipulator round 1:  \(a_1\) gets to pick \(k(k+1) \cdot id(n)\) items. 
    \item Non-manipulators round: 
    \begin{itemize}
        \item Agents in $\AgentSet \setminus \{a_1,d\}$ pick in \(id(n)\) subrounds. In each subround, each of them picks exactly one item in the following order: agents \(c_{j}\), $1 \leq j \leq k$, are the first pickers, then come agents $p_{j,r}$,  
        and lastly agents   \(\overline{c}_j\). 
        \item Finally, agent $d$ picks $k(k+1) \cdot id(n)$ items.
    \end{itemize}
    \item Manipulator round 2:  \(a_1\) gets all remaining items.
\end{itemize}

\emph{Utility values of \(a_1\)}: For ease of presentation, we use two simplifying assumptions. First, we act as if different items can have the same utility value for $a_1$. This assumption can be removed making preferences strict by adding sufficiently small \(\epsilon\) values. Second, we use negative utilities. In fact, one can recover an equivalent instance with only non-negative values by adding to all the utilities the absolute value of the minimal one. Indeed, this would not change the set of optimal solutions as the size of \(a_1\)'s bundle is fixed by \(\pi\).
\begin{itemize}
    \item Items in $Z$ have a utility value of $0$.
    \item One specific item in each set $B_j$, that we denote by $b_j^*$, has a utility value of $4\alpha$ where $\alpha = (id(n) + 2)k(k+1)$. All items in $(\bigcup_{j=1}^k B_j \cup D)\setminus\{b_1^*,\ldots,b_k^*\} $ have a utility value equal to $2\alpha$. 
    \item The utilities of the items in the sets $Id_j$ and $Id_{\overline{j}}$ are defined as follows. Index the items in $Id_j$ (resp. $Id_{\overline{j}}$) from $1$ to $id(n)+2$ according to the preference order of agent $c_j$ (resp. $\overline{c}_j$). Furthermore, let $\Tau_{j} = \{id(i)| \phi(v_i) = j\}$, $\tau_{j}(t)$ denote the $t^{th}$ smallest value in $\Tau_{j}$, and $T_{j} = |\Tau_{j}|$. We also set $\tau_{j}(0) = 0$ and $\tau_{j}(T_{j} + 1) = id(n)+2$. Then, all items receive a utility value of $1$, except for the items of indices $\tau_{j}(t)$ for $t\in \{1,\ldots,T_{j}+1\}$, that get utility $\tau_{j}(t-1) - \tau_{j}(t) + 1$. Notice that, for every $t$ such that $1 \leq t \leq T_{j} + 1$, by definition the sum of the utilities of all the items from 
    $\tau_{j}(t-1)+1$ to $\tau_{j}(t)$ is $0$.

    \item Similarly, the utilities of the items in each $Id_{\{j,r\}}$ are set in the following manner. Index these items from $1$ to $2id(n) + 2$ according to the preference order of agents $p_{j,r}$ and $p_{r,j}$. Furthermore, let $\Tau_{j,r} = \{id(i) +id(l)| \phi(v_i) = j, \phi(v_l) = r, \{v_i,v_l\} \in E\}$, $\tau_{j,r}(t)$ denote the $t^{th}$ smallest value in $\Tau_{j,r}$, and $T_{j,r} = |\Tau_{j,r}|$. We also set $\tau_{j,r}(0) = 0$ and $\tau_{j,r}(T_{j,r} + 1) = 2id(n)+2$. Then all items receive a utility value of $1$, except items of index $\tau_{j,r}(t)$ for $t\in \{1,\ldots,T_{j,r}+1\}$, whose utility is set to $\tau_{j,r}(t-1) - \tau_{j,r}(t) + 1$.
\end{itemize}

As we are going to show below, in an optimal manipulation, the agents behave as follows. In the first manipulator round, $a_1$ picks $k(k+1)\cdot id(n)$ items in $\bigcup_j B_j \cup D$. Then, in the non-manipulators round, agents $c_j$, $p_{j,r}$ and $\overline{c}_j$ for the different values $j$ and $r \neq j$ pick the remaining items in the sets $B_j$, plus other items in $Id_j$, $Id_{\overline{j}}$ and $Id_{\{j,r\}}$. Subsequently, $d$ takes all the remaining items in $D$ and further ones to complete her picks in $Z$. Finally, in the second manipulator round, $a_1$ collects all remaining items.\\ 

\emph{Sketch of the proof.} We first claim that, in the first manipulator round, \(a_1\) should pick only items in $\bigcup_j B_j \cup D$. 
Indeed, after the non-manipulators round, none of these items is left, whatever \(a_1\) has previously picked. In particular, the items left by $a_1$ in each set $B_j$ are collected by agents $c_j$, $p_{j,r}$ and \(\overline{c}_j\), while the ones in $D$ are collected by agent $d$. Moreover, because $|\bigcup_j (Id_j\cup Id_{\overline{j}}) \cup \bigcup_{j\neq r} Id_{j,r}|$ is upper bounded by $\alpha$ and of the utility function we have set, any subset of items in $\ItemSet \setminus (\bigcup_j B_j \cup D)$ as a utility value which is strictly less than $\alpha$ and strictly greater than $-\alpha$. As a result, because each item in $\bigcup_j B_j \cup D$ is worth $2\alpha$, any solution which would not pick only items in $\bigcup_j B_j \cup D$ in the first manipulator round could be improved by doing so.   
Using the same type of argument, we also claim that $a_1$ should pick all of the $b_j^*$ items in her first picking round. We will hence restrict our attention to picking strategies that verify these two assumptions. Under such an hypothesis, the best utility value $a_1$ can hope to get from the set of items she collects in her second picking round is $0$. This is induced by the utility values that we have set, as well as by the truthful picking strategies of non-manipulators. Indeed, note that by construction the overall utility of the set $\ItemSet \setminus (\bigcup_j B_j \cup D)$ is 0. Moreover, as sets $Id_j$, $Id_{\overline{j}}$ and $Id_{\{j,r\}}$ are indexed according to the preference orders of agents $c_j$, $\overline{c}_j$, $p_{j,r}$ and $p_{r,j}$, at the end of the non-manipulators round only prefixes of such sets have been picked. Hence, recalling that all the items in $Z$ have null utility for $a_1$, the overall utility of items left to $a_1$ at the beginning of the second manipulator round is $0$ if and only if the prefixes of the already picked items in all the sets $Id_j$, $Id_{\overline{j}}$ and $Id_{\{j,r\}}$ end up to items of negative value for $a_1$. 
We now argue that this can happen if and only if there exists a multicolored clique $G$.

Let us first show the only if direction, i.e., that if $a_1$ gets an overall utility equal to $0$ from the set of items she collects in her second picking round, then there is a multicolored clique of size $k$ in $G$. Let us denote by $nb_j$ the number of items that $a_1$ has picked in $B_j$ during the first manipulator round. Since for each $j\in \{1,\ldots,k\}$ agent $a_1$ has picked $b_j^*$ and $|B_j|=(k+1)id(n)$, we have that $1\le nb_j \le (k+1)id(n)$. We first show that $nb_j$ should be a multiple of $k+1$. 

To this aim, let us first observe that, for each $j\in \{1,\ldots,k\}$, after the first manipulator round, in every non-manipulators subround, $k+1$ items of $B_j$ (if still available) are picked by the $k+1$ agents $c_j$, $p_{j,r}$ with $j\neq r$ and $\overline{c}_j$ (in this order).
Therefore, at the end of the non-manipulators rounds,  $c_j$ has picked $\lfloor nb_j/(k+1)\rfloor$ items in $Id_{j}$ and $\overline{c}_j$ has picked $\lceil nb_j/(k+1)\rceil$ items in $Id_{\overline{j}}$. But then, if $nb_j$ is not a multiple of $k+1$, these two numbers are different and thus the last items picked by $c_j$ in $Id(j)$ and by $\overline{c}_j$ in $Id_{\overline{j}}$ cannot both have negative utility for $a_1$, because the difference between two consecutive $id$ values is strictly greater than $1$. Therefore, each $nb_j$ should be of the form $nb_j = (k+1) \cdot id(i_j)$ for some $i_j\in \{1,\ldots,n\}$ such that $\phi(v_{i_j}) = j$, so that both $c_j$ and $\overline{c}_j$ pick $id(i_j)$ items in $Id_j$ and $Id_{\overline{j}}$, respectively.
In order to show that $\{v_{i_j}|1\le j\le k\}$ is a multicolored clique, it remains to prove that all the vertices of this set are neighbors in $G$. Indeed, since in each subround of the non-manipulators round every time $c_j$ picks in $Id(j)$ each agent $p_{j,r}$ picks in $Id_{\{j,r\}}$, 
at the end of the non-manipulators round \(p_{j,r}\) and \(p_{r,j}\) have picked $id(i_j)+id(i_r)$ items in $Id_{\{j,r\}}$. Since in order for $a_1$ to achieve an overall utility equal to $0$ in the second manipulator round the last item previously picked in $Id_{\{j,r\}}$ must have a negative utility, \(\{v_{i_j}, v_{i_r}\}\) must be an edge of \(G\).   

It remains to show the if direction, i.e., that if there is a multicolored clique $\{v_{i_1},\ldots,v_{i_k}\}$ in $G$, then there exists a strategy leading $a_1$ to reach overall utility $0$ in her second manipulation round. Assuming without loss of generality that $\phi(v_{l_j}) = j$, this can be accomplished by letting $a_1$ pick $nb_j = (k+1) \cdot id(i_j)$ items in $B_j$, $1 \leq j \leq k$, and the remaining items in $D$. Then, each $c_j$ (resp. $\overline{c}_j$) will pick
$id(i_j)$ items in $Id(j)$ (resp. $Id_{\overline{j}}$) and each $p_{j,r}$ will pick $id(i_j)$ items in $Id_{\{j,r\}}$, which causes $a_1$ to achieve overall utility $0$ in her second manipulator round, finally proving the claim.
\end{proof}

Consequently from Theorems~\ref{theorem: hardness} and \ref{theorem: hardness on n}, it is unlikely that the \(\MSA\) problem admits FPT algorithms w.r.t. parameters \(\mu(a_1)\) and \(n\). Hence, these results valorize the XP results on these parameters obtained in Section~\ref{Sect:positive}, as well as Theorem~\ref{theorem n + mu}, which interestingly shows that the \(\MSA\) problem is FPT w.r.t. parameter \(\mu(a_1) + n\).

\section{An Upper Bound on the Optimal Value of \(\MSA\)} \label{section : bound}
Our initial hope was that computational complexity could be a barrier to manipulating sequential allocation. Unfortunately, we have seen in Section~\ref{Sect:positive} that the \(\MSA\) problem can be solved efficiently for several subclasses of instances. Another reason that could push agents towards behaving truthfully could be that it would not be worth it. Indeed, if the increase in utility that an agent can get by manipulating is very low, she might be reluctant to gather the necessary information and do the effort of looking for a good manipulation.  We provide the following tight bound on this issue.

\begin{theorem} \label{theorem : bound}
The manipulator cannot increase her welfare by a factor greater than or equal to 2, i.e., \(\max_{\succ} u(\AllocationFunction_{\succ}(a_1)) < 2 u_T \)  and this bound is tight.
\end{theorem}
\begin{proof} 
We proceed by induction on the value of parameter \(\mu(a_1)\). If \(\mu(a_1) = 1\), the bound is obvious because \(a_1\) cannot manipulate. If \(\mu(a_1) = 2\), the bound is also easy to prove because \(a_1\) will obtain only two items and the utility of each of them cannot be greater than the one of the first item \(a_1\) picks when behaving truthfully, one having a strictly lower utility. Note that \(u_T > 0\) when \(\mu(a_1) \ge 2\). Let us assume the bound true up to \(\mu(a_1) = k-1\) and let us further assume for the sake of contradiction that there exists an instance \(\mathcal{J}\) with \(\mu(a_1) = k\) where \(\max_{\succ} u(\AllocationFunction_{\succ}(a_1)) \ge 2 u_T \). Moreover, let us denote by \(x_1, \ldots, x_k\) (resp. \(y_1,\ldots, y_k\)) the items picked by \(a_1\) when behaving truthfully (resp. according to one of her best manipulation that we denote by \(\succ_b\)), where the items are ordered w.r.t. the time step at which they are picked (e.g., \(x_1\) is picked first when \(a_1\) behaves truthfully). Our hypothesis implies that:
\begin{equation} \label{hyp}
    \sum_{i=1}^k u(y_i) \ge 2\sum_{i=1}^k u(x_i) =  2 u_T.
\end{equation}

We will now show how to build from \(\mathcal{J}\) an instance \(\mathcal{J}''\) with \(\mu(a_1) = k-1\) and where \(\max_{\succ} u(\AllocationFunction_{\succ}(a_1)) \ge 2 u_T \), hence bringing a contradiction.  However, for ease of presentation this construction is decomposed into two parts: a first one where we work on an instance \(\mathcal{J}'\) obtained from \(\mathcal{J}\); and a second one where we analyse the desired instance \(\mathcal{J}''\) which is obtained from \(\mathcal{J}'\).\\

\textbf{Part 1:} Consider the instance \(\mathcal{J}'\) obtained from \(\mathcal{J}\) by removing the first occurrence of \(a_1\) in \(\PickingSequence\) (an arbitrary non-manipulator is added at the end of \(\PickingSequence\) so that the length of \(\PickingSequence\) remains \(m\)). We denote by \(t^1\) this particular time step, i.e., the one of the first occurrence of \(a_1\) in \(\PickingSequence\). We point out that in \(\mathcal{J}'\), the manipulator can manipulate to obtain the set of items \(y_2,\ldots,y_k\). Indeed, assume w.l.o.g. that items \(y_1,\ldots,y_k\) are ranked first in \(\succ_b\) (not necessarily in that order) and let \(\succ_b^{\downarrow y_1}\) be the ranking obtained from \(\succ_b\) by putting \(y_1\) in last position. Furthermore, let \(S^{\mathcal{J}}_t\) (resp. \(S^{\mathcal{J}'}_t\)) be the set of items picked at the end of time step \(t\) in instances \(\mathcal{J}\) (resp. \(\mathcal{J}'\)) when \(a_1\) follows strategy \(\succ_b\) (resp. \(\succ_b^{\downarrow y_1}\)). Then we have the following lemma.
\begin{lemma} \label{lemma: proof bound 1}
\(S^{\mathcal{J}'}_t = S^{\mathcal{J}}_t\) for \(t < t^1\) and \(S^{\mathcal{J}'}_t \subseteq S^{\mathcal{J}}_{t+1}\) for \(t \ge t^1\).
\end{lemma}
\begin{proof}[Proof of Lemma~\ref{lemma: proof bound 1}]
The first part of the lemma is obvious because the picking processes in \(\mathcal{J}\) and \(\mathcal{J}'\) are identical for \(t < t^1\). 
We prove the second part of the lemma by induction. 
Let us denote by \(\PickingSequence_{\mathcal{J}}(t)\) (resp. \(\PickingSequence_{\mathcal{J}'}(t)\)) the picker at time step \(t\) in instance \(\mathcal{J}\) (resp. \(\mathcal{J}'\)). 
Then for all \(t \ge t^1\), we have that \(\PickingSequence_{\mathcal{J}'}(t) = \PickingSequence_{\mathcal{J}}(t+1)\). 
Then, \(S^{\mathcal{J}'}_t \subseteq S^{\mathcal{J}}_{t+1}\) is true for \(t = t^1\) because \(S^{\mathcal{J}'}_{t^1-1} = S^{\mathcal{J}}_{t^1-1}\), \(S^{\mathcal{J}}_{t^1} = S^{\mathcal{J}}_{t^1-1} \cup \{y_1\}\), and hence the only item that \(\PickingSequence_{\mathcal{J}'}(t^1)\) could prefer to the one \(\PickingSequence_{\mathcal{J}}(t^{1}+1)\) picks is \(y_1\). Let us assume the inclusion relationship true up to \(t = l\). So \(S^{\mathcal{J}}_{l+1} = S^{\mathcal{J}'}_{l} \cup \{i\}\), where \(i\) is some item. Then, we obtain that \(S^{\mathcal{J}'}_{l+1} \subseteq S^{\mathcal{J}}_{l+2}\), because the only item that \(\PickingSequence_{\mathcal{J}'}(l+1)\) could prefer to the one \(\PickingSequence_{\mathcal{J}}(l+2)\) picks is item \(i\).
\end{proof}
A direct consequence of Lemma~~\ref{lemma: proof bound 1}, \(\succ_b^{\downarrow y_1}\) is successful in picking items \(y_2,\ldots,y_k\) in \(\mathcal{J}'\).\\

\textbf{Part 2:} Let us now consider the instance \(\mathcal{J}''\) obtained from \(\mathcal{J}'\) by removing \(x_1\) from the set of items as well as the last agent in the picking sequence (that we had artificially added in the first part of the proof). In \(\mathcal{J}''\) the set of items that \(a_1\) gets when behaving truthfully is \(x_2,\ldots, x_k\). Consider the preference ranking \(\succ_b^{\downarrow y_1, - x_1}\) obtained from \(\succ_b^{\downarrow y_1}\) by removing \(x_1\) and denote by \(S^{\mathcal{J}''}_t\) the set of items picked at the end of time step \(t\) in \(\mathcal{J}''\) when \(a_1\) follows strategy \(\succ_b^{\downarrow y_1, - x_0}\).  
Moreover, let \(t^l\) (resp. \(t^1\)) be the time step in \(\mathcal{J}'\) at which \(a_1\) picks \(y_l\) (resp. at which \(x_1\) is picked by some agent) when \(a_1\) uses strategy  \(\succ_b^{\downarrow y_1}\) where \(2 \leq l \leq k\).  
We show that \(a_1\) can get in \(\mathcal{J}''\) a set of items \(Y\) compounded of all items in \(\{y_2,\ldots, y_k\}\) up to one item. 
This is a consequence of the following lemma.
\begin{lemma}\label{lemma: proof bound 2}
\(S^{\mathcal{J}''}_t = S^{\mathcal{J}'}_t\) for \(t < t^1\) and \(S^{\mathcal{J}''}_t\) is of the form \((S^{\mathcal{J}'}_t\setminus \{x_1\}) \cup \{i\}\) for \(t \ge t^1\) where \(i\) is some item in \(\ItemSet\).
\end{lemma}
\begin{proof}[Proof of Lemma~\ref{lemma: proof bound 2}]
The first part of the lemma is obvious because the picking processes in \(\mathcal{J}'\) and \(\mathcal{J}''\) are identical for \(t < t^1\). 
We prove the second part of the lemma by induction. 
Let us denote by \(\PickingSequence_{\mathcal{J}'}(t)\) (resp. \(\PickingSequence_{\mathcal{J}''}(t)\)) the picker at time step \(t\) in instance \(\mathcal{J}'\) (resp. \(\mathcal{J}''\)). 
Then for all \(t\), we have that \(\PickingSequence_{\mathcal{J}'}(t) = \PickingSequence_{\mathcal{J}''}(t)\). 
Then, the fact that \(S^{\mathcal{J}''}_{t^1}\) is of the form \((S^{\mathcal{J}'}_{t^1}\setminus \{x_1\}) \cup \{i\}\)  is just due to the fact that \(S^{\mathcal{J}'}_{t^1-1} = S^{\mathcal{J}''}_{t^1-1}\). 
Let us assume this fact true up to \(t = l\). So \(S^{\mathcal{J}''}_{l} = S^{\mathcal{J}'}_{l} \setminus \{x_1\} \cup \{i\}\), where \(i\) is some item. We obtain that \(S^{\mathcal{J}''}_{l+1}\) is of the form \(S^{\mathcal{J}'}_{l} \setminus \{x_1\} \cup \{i'\}\), where \(i'\) is some item, because \(\PickingSequence_{\mathcal{J}''}(l+1)\) will pick the same item as \(\PickingSequence_{\mathcal{J}'}(l+1)\) except if this item is \(i\).
\end{proof}
Let \(j\) be the first index in \(\mathcal{J}''\) (if any) such that \(a_1\) cannot pick \(y_j\) at \(t^j \ge t^1\). Then, by Lemma \ref{lemma: proof bound 2},
we have that \(S^{\mathcal{J}''}_t = (S^{\mathcal{J}'}_t\setminus \{x_1\}) \cup \{y_{l+1}\}\) for \(t^j \le t^l \le t < t^{l+1}\). This is due to the fact that \(\succ_b^{\downarrow y_1}\) is successful in picking items \(y_2,\ldots,y_k\) in \(\mathcal{J}'\) and this proves the claim that \(a_1\) can get a set of items \(Y\) compounded of all items in \(\{y_2,\ldots, y_k\}\) up to one item. 
Hence, 
\begin{align*}
\sum_{y \in Y } u(y) &\ge \sum_{i = 1}^k u(y_i) - u(y_1) - \max_{2\le j\le k} u(y_j)\\
&\ge 2\sum_{i=1}^k u(x_i) - u(y_1) - \max_{2\le j\le k} u(y_j) \\
&\ge 2\sum_{i=2}^k u(x_i) 
\end{align*}
where the second inequality is due to Inequality \ref{hyp} and the third one is due to fact that \(u(y_i) \le u(x_1), \forall i \in \{1,\ldots,k\}\). We obtain a contradiction because \(\mu(a_1) = k-1\) in \(\mathcal{J}''\).

The tightness of the bound is provided by the instance of Example \ref{example:running} with the following utility function \(u(i_1) = u(i_2) + \epsilon = u(i_3) + 2 \epsilon = 1\) and \(u(i_4) = 0\). We recall that in this instance \(u_T = u(i_1) + u(i_4) = u(i_1)\) whereas the manipulator can obtain by manipulating the set \(S=\{i_2,i_3\}\) with utility \(u(i_2) + u(i_3) = 2u(i_1) -3\epsilon\). 
\end{proof}
We conclude from Theorem \ref{theorem : bound} that, while the increase in utility of the manipulator cannot be arbitrarily large, manipulating may often be worth it for the manipulator as doubling her utility can be a significant improvement.
\section{An Integer Programming Formulation}\label{Sect:IntProg}

In this last section, we provide an integer programming formulation of the \(\MSA\) problem. 
This integer program provides another tool than the dynamic programming algorithm of Section~\ref{Sect:positive} to solve the \(\MSA\) problem which may be more efficient for some instances. Moreover, a more thorough analysis of this formulation and of it's solution polytope may lead to new results. For instance, some bounds on the number of variables could yield new parameterized complexity results via Lenstra's theorem \cite{lenstra1983integer}.

By abuse of notation, we will identify the set \(\ItemSet\) as the set \([m] = \{1,\ldots,m\}\). Let \(x_{it}\) be a binary variable with the following meaning: \(x_{it} = 1\) iff item \(i\) is picked at time step \(t\). Then, the constraint \(\sum_{i = 1}^m x_{it} = 1\) implies that there is exactly one item being picked at time step \(t\). Similarly, the constraint \(\sum_{t = 1}^m x_{it} = 1\) implies that item \(i\) is picked at exactly one time step. 
Lastly, consider a time step \(t\) such that \(\PickingSequence(t) \neq a_1\). Then the following set of constraints implies that \(\PickingSequence(t)\) picks at time step \(t\) the best item available for her.
\begin{equation}
x_{it} + \sum_{j \succ_{\PickingSequence(t)} i} x_{jt} + \sum_{t' < t} x_{it'} \geq 1, \quad \forall i \in [m] \notag
\end{equation}
In words, this constraint says that: if \(\PickingSequence(t)\) does not pick item \(i\) at time step \(t\) (i.e., \(x_{it} = 0\)), it is either because she picks something that she prefers (i.e., \(\sum_{j \succ_{\PickingSequence(t)} i} x_{jt} = 1\) ), or because item \(i\) has already been picked at an earlier time step (i.e., \(\sum_{t' < t} x_{it'} = 1\)). To summarize, we obtain the following integer programming formulation\footnote{We note that we obtain an integer programming formulation which is very close to one of the mathematical program used to solve stable marriage problems \cite{gusfield1989stable}.} with \(m^2\) binary variables and at most \(m^2\) constraints (as we can assume the manipulator is the picker for at least two time steps), 
\begin{align*}
    \max_{x_{it}} \sum_{i=1}^m\sum_{t : \PickingSequence(t) = a_1} x_{it} u(i)& \\
    \text{s.t. \hspace{1.cm}}    \sum_{t = 1}^m x_{it} &= 1, \quad \forall i \in [m]\\
    \sum_{i = 1}^m x_{it} &= 1, \quad \forall t \in [m]\\
    x_{it} + \sum_{j \succ_{\PickingSequence(t)} i} x_{jt} + \sum_{t' < t} x_{it'} &\geq 1, \quad \forall i,t \in [m]^2 \text{ s.t. } \PickingSequence(t) \neq a_1 \\
    x_{it} &\in \{0,1\}, \quad \forall i,t \in [m]^2 
\end{align*}

\section{Conclusion and Future Work}

We have provided a variety of results on the problem of finding an optimal manipulation in the sequential allocation protocol. Beside an integer program to solve this problem, we have designed a dynamic programming algorithm from which we have derived several positive parameterized complexity results. For instance, we have shown that this manipulation problem is in XP with respect to the number of agents and that it is in FPT with respect to the maximum range of an item. Conversely, we have also provided matching W[1]-hardness results. 
Lastly, motivated by the fact that agents could be inclined to behave truthfully if a manipulation would not be worth it, we have investigated an upper bound on the increase in utility that the manipulator could get by manipulating. We have showed that the manipulator cannot increase the utility of her bundle by a factor greater than or equal to 2 and that this bound is tight. Overall, our results show that, not only sequential allocations are worth manipulating, but also that they can be manipulated efficiently for a wide range of instances.

Several directions for future work are conceivable. One could try to decrease our upper bound on the increase in utility that the manipulator can obtain by restricting to specific instances (e.g., imposing a specific type of picking sequences). Moreover, it would be worth investigating the price of manipulation, i.e., the worst case ratio between the social welfare when one agent manipulates, all the others being truthful, and the one when all the agents behave truthfully. On this issue, to the best of our knowledge, little is known except for some preliminary results by Bouveret and Lang~\cite{bouveret2014manipulating}. \\

\textbf{Acknowledgments.} We are grateful to J\'er\^ome Lang and Paolo Serafino for helpful comments and stimulating conversations on this work.

\bibliographystyle{alpha}
\bibliography{references}

\end{document}